\makeatletter
\newcommand{\dontusepackage}[2][]{%
  \@namedef{ver@#2.sty}{9999/12/31}%
  \@namedef{opt@#2.sty}{#1}}
\makeatother
\dontusepackage{subfigure}

\documentclass[]{article}

\usepackage{lmodern}
\usepackage{amssymb,amsmath}
\usepackage{amsthm}
\usepackage{ifxetex,ifluatex}
\usepackage[usenames,dvipsnames]{color}
\usepackage{fixltx2e} 
\ifnum 0\ifxetex 1\fi\ifluatex 1\fi=0 
  \usepackage[T1]{fontenc}
  \usepackage[utf8]{inputenc}
\else 
  \ifxetex
    \usepackage{mathspec}
    \usepackage{xltxtra,xunicode}
  \else
    \usepackage{fontspec}
  \fi
  \defaultfontfeatures{Mapping=tex-text,Scale=MatchLowercase}
  
\fi
\IfFileExists{upquote.sty}{\usepackage{upquote}}{}
\IfFileExists{microtype.sty}{%
\usepackage{microtype}
\UseMicrotypeSet[protrusion]{basicmath} 
}{}
\usepackage[margin=1.0in,bottom=1.5in]{geometry}
\usepackage[]{natbib}
\usepackage{listings}
\usepackage{graphicx}
\makeatletter
\def\maxwidth{\ifdim\Gin@nat@width>\linewidth\linewidth\else\Gin@nat@width\fi}
\def\maxheight{\ifdim\Gin@nat@height>\textheight\textheight\else\Gin@nat@height\fi}
\makeatother
\setkeys{Gin}{width=\maxwidth,height=\maxheight,keepaspectratio}
\usepackage{caption}
\usepackage{float}
\usepackage{subfig}
\ifxetex
  \usepackage[setpagesize=false, 
              unicode=false, 
              xetex]{hyperref}
\else
  \usepackage[unicode=true]{hyperref}
\fi
\hypersetup{breaklinks=true,
            bookmarks=true,
            pdfauthor={},
            pdftitle={Minkowski sets for the regularization of inverse problems.},
            colorlinks=true,
            citecolor=black,
            urlcolor=blue,
            linkcolor=black,
            pdfborder={0 0 0}}
\def\argmin{\mathop{\rm arg\min}}

\theoremstyle{definition}
\newtheorem{mydef}{Definition}

\theoremstyle{plain}
\newtheorem{myprop}{Proposition}

\usepackage{authblk}

\begin{document}
\title{Generalized Minkowski sets for the regularization of inverse problems.}
\author[1]{Bas Peters}
\author[2]{Felix J. Herrmann}

\affil[1]{University of British Columbia, Vancouver, Canada}
\affil[2]{Georgia Institute of Technology.}
\date{}
\maketitle

\begin{abstract}
Many works on inverse problems in the imaging sciences consider regularization via one or more penalty functions or constraint sets. When the models/images are not easily described using one or a few penalty functions/constraints, additive model descriptions for regularization lead to better imaging results. These include cartoon-texture decomposition, morphological component analysis, and robust principal component analysis; methods that typically rely on penalty functions. We propose a regularization framework, based on the Minkowski set, that merges the strengths of additive models and constrained formulations. We generalize the Minkowski set, such that the model parameters are the sum of two components, each of which is constrained to an intersection of sets. Furthermore, the sum of the components is also an element of another intersection of sets. These generalizations allow us to include multiple pieces of prior knowledge on each of the components, as well as on the sum of components, which is necessary to ensure physical feasibility of partial-differential-equation based parameters estimation problems. We derive the projection operation onto the generalized Minkowski sets and construct an algorithm based on the alternating direction method of multipliers. We illustrate how we benefit from using more prior knowledge in the form of the generalized Minkowski set using seismic waveform inversion and video background-anomaly separation.
\end{abstract}

\section{Introduction}\label{introduction}

The inspiration for this work is twofold. First, we want to build on the success of the regularization of inverse problems in the imaging sciences using intersections of constraint sets. This approach may be in the form of minimizing a differentiable function $f(m) : \mathbb{R}^N \rightarrow \mathbb{R}$, such that the estimated model parameters, $m \in \mathbb{R}^N$ are an element of the intersection of $p$ convex and possibly non-convex sets, $\bigcap_{i=1}^p \mathcal{C}_i$. The corresponding optimization problem reads $\min_m f(m) \:\: \text{s.t.} \:\: m \in \bigcap_{i=1}^p \mathcal{C}_i$. See, e.g., \citep{CMRseismic, Cseismicrecovery, PDTAF, smithyman2015constrained, doi:10.1190/tle36010094.1, Esser2016arch, TVWRI2, ournewpreprint}. Alternatively, we can also formulate certain types of inverse problems directly as a feasibility (also known as set-theoretic estimation) or projection problem \citep{STE_2, STE_1, STEstimation, COMBETTES1996155}. This approach solves $\min_m 1/2 \| x-m \|_2^2 \:\: \text{s.t.} \:\: m \in \bigcap_{i=1}^p \mathcal{C}_i$ or seeks any $m \in \bigcap_{i=1}^p \mathcal{C}_i$. The mentioned regularization techniques describe each piece of prior knowledge about the model parameters $m$ as a set. However, prior knowledge represented as a single set or as an intersection of different sets may not capture all we know. For instance, if the model contains oscillatory as well as blocky features. Because these are fundamentally different properties, working with one or multiple constraint sets alone may not able to express the simplicity of the entire model. 

Motivated by ideas from cartoon-texture decomposition (CTD) and morphological component analysis (MCA) \citep{doi:10.1137/S1540345902416247, doi:10.1137/110854989,cta_bpdn,cta_rank} and robust or sparse principal component analysis (RPCA) \citep{Candes:2011:RPC:1970392.1970395,rpca_td1,cta_td1}, we consider an additive model structure. Additive model structures decompose $m$ into two or more components---e.g., into a blocky cartoon-like component $u \in \mathbb{R}^N$ and an oscillatory texture-like component $v \in \mathbb{R}^N$. For this purpose, a penalty method is often used to state the decomposition problem as
\begin{equation}
\min_{u,v} \| m - u -v \| + \frac{\alpha}{2} \| A u \| + \frac{\beta}{2} \| B v \|.
\label{CTD}
\end{equation}
While this method has been successful, it requires careful choices for the penalty parameters $\alpha >0$ and $\beta >0$. These parameters determine how `much' of $m$ ends up in each component, but also depend on the noise level. In addition, the values of these parameters relate to the choices for the norms and the linear operators $A \in \mathbb{C}^{M_1 \times N}$ and
$B \in \mathbb{C}^{M_2 \times N}$. In this work, we are concerned with regularizing inverse problems by including multiple pieces of prior knowledge. For additive models, this means the number of terms in the sum and thus the number of penalty parameters will increase rapidly.

A constrained formulation can partially avoid the difficulties of selecting numerous penalty parameters because each constraint set can be defined independently of all other sets. This is contrary to penalty methods where the penalty parameters are trade-offs between the different penalty functions. Therefore, we explore the use of Minkowski sets for regularizing inverse problems. We require that the vector of model parameters, $m$, is an element of the Minkowski set $\mathcal{V}$, or vector sum of two sets $\mathcal{C}_1$ and $\mathcal{C}_1$, which is defined as
\begin{equation}
\mathcal{V} \equiv \mathcal{C}_1 + \mathcal{C}_1 = \{ m = u + v \: | \: u \in \mathcal{C}_1, v \in \mathcal{C}_2 \}.
\label{Minkowski_set}
\end{equation}
A vector $m$ is an element of $\mathcal{V}$ if it is the sum of vectors $u \in \mathcal{C}_1$ and $v \in \mathcal{C}_2$. Each set describes particular model properties for each component. Examples are total-variation, sparsity in a transform domain (Fourier, wavelets, curvelets, shearlets) or matrix rank. For practical reasons, we assume that all sets, sums of sets, and intersections of sets are non-empty, which implies that projection and feasibility problems have at least one solution. Our algorithm design uses the property that the sum of $p$ sets $\mathcal{C}_i$ is convex if all $\mathcal{C}_i$ are convex \citep[p. 24]{hiriart2012fundamentals}. We apply our set-based regularization to imaging inverse problems with the Euclidean projection operator:
\begin{equation}
\mathcal{P}_{\mathcal{V}} (m) \in \argmin_{x} \frac{1}{2} \| x - m \|_2^2 \quad \text{s.t} \quad x \in \mathcal{V}. 
\label{proj_Mink_basic}
\end{equation}
This projection allows us to use Minkowski constraint sets in algorithms
such as (spectral) projected gradient \citep[SPG,][]{Birgin:1999:NSP:588891.589081}, projected quasi-Newton
\citep[PQN,][]{schmidt2009optimizing}, projected Newton algorithms
\citep{doi:10.1137/0320018, pnmethods}, and proximal-map based
algorithms if we include the Minkowski constraint as an indicator
function. We define this indicator function for a set $\mathcal{V}$ as
\begin{equation}
    \iota_{\mathcal{V}}(m) =
    \begin{cases}
      0 & \text{if } m \in \mathcal{V},\\
      +\infty & \text{if } m \notin \mathcal{V},\\
\end{cases}
\label{indicator}
\end{equation}
 and the proximal map for a function
$g(m) : \mathbb{R}^N \rightarrow \mathbb{R} \cup \{+\infty\}$ as
$\operatorname{prox}_{\gamma,g}(m) = \argmin_x g(x) + \frac{\gamma}{2}\|x-m\|_2^2$,
with $\gamma >0$. The proximal map for the indicator function of a set
is the projection:
$\operatorname{prox}_{\gamma,\iota_\mathcal{V}}(m)= \mathcal{P}_{\iota_{\mathcal{V}}}(m)$.

While Minkowski sets and additive models are powerful regularizers, they lack certain critical features needed for solving problems that involve physical parameters. For instance, there is, in general, no guarantee that the sum of two or more components lies within lower and upper bounds or satisfies other crucial constraints that ensure physically realistic model estimates. It is also not straightforward to include multiple pieces of prior information on each component.

\subsection{Related work}
The above introduced decomposition strategies, morphological component analysis or cartoon-texture \citep{doi:10.1137/S1540345902416247,doi:10.1137/110854989,cta_bpdn,cta_rank} and robust or sparse principal component analysis \citep{Candes:2011:RPC:1970392.1970395,rpca_td1,cta_td1}), share the additive model construction with multiscale decompositions in image processing \citep[e.g.,][]{meyer2001oscillating, doi:10.1137/030600448}. While each of the sets that appear in a Minkowski sum can describe a particular
scale, this is not our primary aim or motivation. We use the
summation structure to build more complex models out of simpler ones,
more aligned with cartoon-texture decomposition and robust principal
component analysis.

Projections onto Minkowski sets also appear in computational geometry,
collision detection, and computer-aided design
\citep[e.g.,][]{Dobkin1993ComputingTI, VARADHAN2006343, LEE201614}, but
the problems and applications are different. In our case, sets describe
model properties and prior knowledge in $\mathbb{R}^N$. In computational
geometry, sets are often the vertices of polyhedral objects in
$\mathbb{R}^2$ or $\mathbb{R}^3$ and do not come with closed-form expressions for projectors or the Minkowski sum. We do not need to form the Minkowski set explicitly, and we show that projections onto the set are sufficient to regularize inverse problems.

\subsection{Contributions and outline}\label{contributions}

We propose a constrained regularization approach suitable for inverse problems with an emphasis on physical parameter estimation, where we describe the model parameters using multiple components. This implies that we need to work with multiple constraints for each component while offering assurances that the sum of the components also adheres to certain constraints. For this purpose, we introduce a generalization of the Minkowski set and a formulation void of penalty parameters. We formulate the projection on such generalized Minkowski sets and construct an algorithm based on the alternating direction method of multipliers, including a discussion on important algorithmic details. We show how to use the projector to formulate inverse problems based on these sets. 

The algorithms and numerical example are available as an open-source software package written in Julia\footnote{The software is build on top of the \texttt{SetIntersectionProjection} package available at \url{https://github.com/slimgroup/SetIntersectionProjection.jl}}. The algorithms and software are suitable for small 2D models, as well as for larger 3D geological models or videos, as we will show in the numerical examples section using seismic parameter estimation and video processing examples. These examples also demonstrate that the developed constrained problem formulation, algorithm, and software allow us to define constraints based on the entire 2D/3D model, but also simultaneously on slices/rows/columns/fibers of that model. This feature enables us to include more prior knowledge more directly into the inverse problem.

\section{Generalized Minkowski set}\label{Generalized-minkowski-set}

It is challenging to select a single constraint set or intersection of multiple sets to describe models and images that contain distinct morphological components $u$ and $v$. While the Minkowski set allows us to define different constraint sets for the different components, limitations come to light when working with physical parameter estimation applications.

For instance, there is usually prior knowledge about the physically realistic values in $m \in \mathbb{R}^N$. Moreover, the successful applications of constrained formulations that we discussed in the introduction all use multiple constraints on the model parameters, and we want to combine that concept with constraints on the components.

The second extension of the basic concept of a Minkowski set is that we allow the constraint set on each component to be an intersection of multiple sets as well. In this way, we can include multiple pieces of prior information about each component.

We propose the generalized Minkowski constraint set for the regularization of inverse problems as
\begin{mydef}[Generalized Minkowski Set]
A vector $m \in \mathbb{R}^N$ is an element of the generalized Minkowski set,
\begin{equation}
\mathcal{M} \equiv \{ m = u + v \: | \: u \in \bigcap_{i=1}^p \mathcal{D}_i, \: v \in \bigcap_{j=1}^q \mathcal{E}_j, \: m \in \bigcap_{k=1}^r \mathcal{F}_k \},
\label{Minkowski_general}
\end{equation}
if $m$ is an element of the intersection of $r$ sets $\mathcal{F}_k$ and also the sum of two
components $u \in \mathbb{R}^N$ and $v \in \mathbb{R}^N$. The vector $u$
is an element of the intersection of $p$ sets $\mathcal{D}_i$ and $v$ is an
element of the intersection of $q$ sets $\mathcal{E}_j$.
\end{mydef}
The convexity of $\mathcal{M}$ follows from the properties of the sets
$\mathcal{D}_i$, $\mathcal{E}_j$ and $\mathcal{F}_k$. 

\begin{myprop}
The generalized Minkowski set is convex if $\mathcal{D}_i$, $\mathcal{E}_j$ and $\mathcal{F}_k$ are convex sets for all $i$, $j$ and $k$.
\end{myprop}
\begin{proof}
It follows directly from the definition that $\bigcap_{i=1}^p \mathcal{D}_i$, $\bigcap_{j=1}^q \mathcal{E}_j$, and
$\bigcap_{k=1}^r \mathcal{F}_k$ are convex sets. As a result, the Minkowski sum
$\bigcap_{i=1}^p \mathcal{D}_i + \bigcap_{j=1}^q \mathcal{E}_j$ is
is also convex by because it is the sum of two convex sets \citep[p. 24]{hiriart2012fundamentals}. Finally, It follows that $\mathcal{M}$ is a
convex set, because it is the intersection of the convex intersection $\bigcap_{k=1}^r \mathcal{F}_k$, and the convex sum  $\bigcap_{i=1}^p \mathcal{D}_i + \bigcap_{j=1}^q \mathcal{E}_j$. 
\end{proof}

To summarize in words, $m$ is an element of the
intersection of two convex sets, one is the convex Minkowski sum, the
other is a convex intersection. The set $\mathcal{M}$ is therefore also
convex. Note that convexity and
closedness of $\bigcap_{i=1}^p \mathcal{D}_i$ and
$\bigcap_{j=1}^q \mathcal{E}_j$ does not imply their sum is closed. 

It is conceptually straightforward to extend set
definition~\eqref{Minkowski_general} to a sum of three or more components,
but we work with two components for the remainder of this paper for
notational convenience. In the discussion section, we highlight some
potential computational challenges that come with a generalized
Minkowski sets of more than two components.

\section{Projection onto the generalized Minkowski
set}\label{projection-onto-the-generalized-minkowski-set}

In the following section, we show how to use the generalized Minkowski
set, Equation~\eqref{Minkowski_general}, to regularize inverse problems
with computationally cheap or expensive forward operators. First, we
need to develop an algorithm to compute the projection onto
$\mathcal{M}$. Using the projection operator, $\mathcal{P}_{\mathcal{M}}(\cdot)$, we can formulate inverse problems as a
projection directly. Alternatively, we can use the projection operator inside
projected gradient/Newton-type algorithms. Each constraint set
definition may include a linear operator (the transform-domain operator)
in its definition. We make the linear operators explicit, because the projection operator corresponding to, for example,
$\{ x \: | \: \|x\|_2 \leq \sigma \}$, is available in closed form and
easy to compute, but the projection onto
$\{ x \: | \: \|A x\|_2 \leq \sigma \}$ is not when $AA^T \neq \alpha I$
for $\alpha >0$, see examples in, e.g., \cite{prox_split,OPT-003}, \citealp[ ch. 6
\&
7]{doi:10.1137/1.9781611974997}, \cite{doi:10.1080/10556788.2017.1304548}.
Let us introduce the linear operators $A_i \in \mathbb{R}^{M_i \times N}$,
$B_j \in \mathbb{R}^{M_j \times N}$, and
$C_k \in \mathbb{R}^{M_k \times N}$. With indicator functions and
exposed linear operators, we formulate the projection of
$m \in \mathbb{R}^N$ onto set~\eqref{Minkowski_general} as
\begin{equation}
\begin{aligned}
\min_{u,v,w} \frac{1}{2} \| w - m \|_2^2 + \sum_{i=1}^p \iota_{\mathcal{D}_i}(A_i u) + \sum_{j=1}^q \iota_{\mathcal{E}_j}(B_j v) + \sum_{k=1}^r \iota_{\mathcal{F}_k}(C_k w) + \iota_{w=u+v}(w,u,v),
\end{aligned}
\label{proj_Minkowski}
\end{equation}
 where $\iota_{w=u+v}(w,u,v)$ is the indicator function for the equality
constraint $w=u+v$ that occurs in the definition of $\mathcal{M}$. The
sets $\mathcal{D}_i$, $\mathcal{E}_j$, and $\mathcal{F}_k$ have the same
role as in the previous section. The above problem is the minimization
of a sum of functions acting on different as well as shared variables.
We recast it in a standard form such that we can solve it using
algorithms based on the alternating direction method of multipliers \citep[ADMM, e.g.,][]{Boyd:2011:DOS:2185815.2185816,eckstein2015understanding}. Rewriting in a standard form allows
us to benefit from recently proposed schemes for selecting algorithm
parameters that decrease the number of iterations and lead to more
robust algorithms in case we use non-convex sets \citep{empiricalncvxadmm,Xu_2017_CVPR}. As a first step, we
introduce the vector $x \in \mathbb{R}^{2N}$ that stacks two out of the
three optimization variables in~\eqref{proj_Minkowski} as
\begin{equation}
\begin{split}
x \equiv \begin{pmatrix} u \\ v \end{pmatrix}.
\end{split}
\label{x_def}
\end{equation}
 We substitute this new definition in Problem~\eqref{proj_Minkowski} and
eliminate the equality constraints $w=u+v$ to arrive at
\begin{equation}
\begin{aligned}
\min_{x} \frac{1}{2} \| \begin{pmatrix} I_N \:\:\: I_N \end{pmatrix} x - m \|_2^2 + \sum_{i=1}^p \iota_{\mathcal{D}_i}(\begin{pmatrix} A_i \:\:\: 0 \end{pmatrix}x) + \sum_{j=1}^q \iota_{\mathcal{E}_j}(\begin{pmatrix} 0 \:\:\: B_j \end{pmatrix} x) + \sum_{k=1}^r \iota_{\mathcal{F}_k}(\begin{pmatrix} C_k \:\:\: C_k \end{pmatrix} x),
\end{aligned}
\label{proj_Minkowski_2}
\end{equation}
 where $0_{M_i \times N} \in \mathbb{R}^{M_i \times N}$ indicate all
zeros matrices of appropriate dimensions. Next, we take the linear
operators out of the indicator function such that we end up with sub-problems that are usually projections with closed-form solutions. Thereby we avoid the need for nesting iterative algorithms to solve sub-problems related to the indicator functions of constraint sets. 

To separate indicator functions and linear operators, we introduce additional vectors $y_i$ for $i \in \{ 1,2,\dots,p+q+r+1 \}$ of appropriate dimensions. From now on we use $s = p+q+r+1$ to shorten notation. With the new variables,
we rewrite problem formulation~\eqref{proj_Minkowski_2} and add linear
equality constraints to obtain
\begin{equation}
\begin{aligned}
\min_{\{y_i \},x} \frac{1}{2} \| y_s - m \|_2^2 + \sum_{i=1}^p \iota_{\mathcal{D}_i}(y_i) + \sum_{j=1}^q \iota_{\mathcal{E}_j}(y_{j+p}) + \sum_{k=1}^r \iota_{\mathcal{F}_k}(y_{k+p+q}) \:\: \text{s.t.} \:\: \tilde{A} x = \tilde{y},
\end{aligned}
\label{proj_Minkowski_3}
\end{equation}
 where
\begin{equation*}
\begin{aligned}
\tilde{A} x = \tilde{y} \Leftrightarrow
\begin{pmatrix}
A_1 & 0 \\
\vdots & 0 \\
A_p & 0 \\
0 & B_1 \\
0 & \vdots \\
0 & B_q \\
C_1 & C_1 \\
\vdots & \vdots \\
C_r & C_r \\
I_N & I_N \\
\end{pmatrix}
\begin{pmatrix}
u \\
v
\end{pmatrix}
=
\begin{pmatrix}
y_1\\
\vdots\\
y_p\\
y_{p+1}\\
\vdots\\
y_{p+q}\\
y_{p+q+1}\\
\vdots\\
y_{p+q+r}\\
y_{p+q+r+1}\\
\end{pmatrix}.
\end{aligned}
\end{equation*}
 Now define the new function
\begin{equation}
\tilde{f}(\tilde{y},m) \equiv \sum_{i=1}^s f_i(y_i,m) \equiv \frac{1}{2} \| y_s - m \|_2^2 + \sum_{i=1}^p \iota_{\mathcal{D}_i}(y_i) + \sum_{j=1}^q \iota_{\mathcal{E}_j}(y_{j+p}) + \sum_{k=1}^r \iota_{\mathcal{F}_k}(y_{k+p+q}),
\label{f_def}
\end{equation}
 such that we obtain the projection problem in the form
\begin{equation}
\min_{x,\tilde{y}} \tilde{f}(\tilde{y},m) \:\: \text{s.t.} \:\: \tilde{A} x = \tilde{y}.
\label{proj_Minkowski_4}
\end{equation}
 If $x$ and $\tilde{y}$ are a solution to this problem, the equality
constraints enforce $u + v = y_{p+q+r+1}$ and we recover the projection
of $m$ as $y_{p+q+r+1}$ or as
$\begin{pmatrix} I_N \:\: I_N \end{pmatrix} x$. Now that
Problem~\eqref{proj_Minkowski_4} is in a form that we can solve with the ADMM algorithm, we proceed by writing down the augmented
Lagrangian for Problem~\eqref{proj_Minkowski_4} \citep[Chapter
17]{Nocedal:2000} as
\begin{equation}\label{aug_lag}
L_{\rho_1, \dots, \rho_s} (x, y_1, \dots, y_s, v_1, \dots, v_s) = \sum_{i=1}^s \bigg[ f_i(y_i) + v_i^T (y_i - A_i x) + \frac{\rho_i}{2} \| y_i - A_i x \|^2_2 \bigg],
\end{equation}
 where $\rho_i >0$ are the augmented Lagrangian penalty parameters and
$v_i \in \mathbb{R}^{M_i}$ are the vectors of Lagrangian multipliers. We
denote a block-row of the matrix $\tilde{A}$ as $A_i$. The
relaxed ADMM iterations with relaxation parameters $\gamma_i \in (0,2]$
and iteration counter $l$ are given by
\begin{align}\label{PMinkowski_iters}
&x^{l+1} = \arg\min_{x} \sum_{i=1}^s \frac{\rho_i^l}{2} \| y_i^{l} - A_i x + \frac{v_i^l}{\rho_i^l} \|^2_2 = \Big[ \sum_{i=1}^{s} ( \rho_i^l A_i^T A_i ) \Big]^{-1} \sum_{i=1}^s \Big( A_i^T(\rho_i^l y_i^{l} + v_i^l) \Big) \nonumber \\
&\text{compute for} \:\: i \in  \{1,\dots,s\} \:\: \text{independently in parallel:}\nonumber \\ 
&\bar{x}_i^{l+1} = \gamma_i^l A_i x_i^{l+1} + ( 1-\gamma_i^l ) y_i^{l}\\
&y_i^{l+1} = \arg\min_{y_i} \Big[ f_i(y_i) + \frac{\rho_i^l}{2} \| y_i^{l} - \bar{x}_i^{l+1} + \frac{v_i^l}{\rho_i^l} \|^2_2 \Big] = \operatorname{prox}_{f_i,\rho_i}(\bar{x}_i^{l+1} - \frac{v_i^l}{\rho_i^l} )\nonumber\\ 
&v_i^{l+1} = v_i^l + \rho_i^l (y_i^{l+1} - \bar{x}_i^{l+1}). \nonumber
\end{align}

These iterations are similar to the Simultaneous Direction Method of
Multipliers \citep[SDMM,][]{prox_split,CoSpSDMM}) and the SALSA
algorithm \citep{SalsaPaper}. The difference is that we have an additional
relaxation step and use a separate augmented Lagrangian penalty parameter and relaxation parameter for each index $i$ that corresponds to a different set. The structure of these iterations look identical to the algorithm
presented by \cite{peters2019algorithms} to compute the projection
onto an intersection of sets, but here we solve a different problem and
have different block-matrix structures. We briefly mention the main properties
of each sub-problem.

\textbf{$x^{l+1}$ computation}. This step is the solution of a large,
sparse, square, and symmetric linear system. The
system matrix has the following block structure:
\begin{equation}
\begin{aligned}
&Q \equiv \sum_{i=1}^{s} ( \rho_i A_i^T A_i ) =\\
&\begin{pmatrix} \sum_{i=1}^p \rho_i A_i^T A_i + \sum_{k=1}^r \rho_k C_k^T C_k + \rho_s I_N & \sum_{k=1}^r \rho_k C_k^T C_k + \rho_s I_N \\ \sum_{k=1}^r \rho_k C_k^T C_k + \rho_s I_N & \sum_{j=1}^q \rho_j B_j^T B_j + \sum_{k=1}^r \rho_k C_k^T C_k + \rho_s I_N \end{pmatrix}.
\end{aligned}
\label{system-mat}
\end{equation}
 This matrix is positive-definite if $\tilde{A}$ has full
column rank. We assume this is true in the remainder because we always use bound constraints on each of the components. This implies that at least one of the $A_i$ and one of the $B_i$ is the identity matrix, which makes $Q$ positive-definite. We compute $x^{l+1}$ with the
conjugate gradient (CG) method, warm started by $x^l$ as the initial guess.
We choose CG instead of an iterative method for least-squares problems
such as LSQR \citep{Paige:1982:LAS:355984.355989} because solvers for
least-squares work with $A$ and $A^T$ separately and
need to compute a matrix-vector product (MVP) with each $A_i$ and
$A_i^T$ at every iteration. This becomes computationally
expensive if there are many linear operators, as is the case for our
problem. CG uses a single MVP with $Q$ per iteration. The cost of this
MVP does not increase if we add orthogonal matrices to $\tilde{A}$ (DCT, DFT, various wavelet transforms). If
the matrices in $\tilde{A}$ have (partially) overlapping sparsity
patterns (matrices based on discrete derivatives for exmple), the cost also does not increase (much). We pre-compute all
$A_i^T A_i$ for fast updating of $Q$ when one or more of
the $\rho_i$ change (see below).

\textbf{$y_i^{l+1}$ computation}. For every index $i$, we can compute
$\operatorname{prox}_{f_i,\rho_i}(\bar{x}_i^{l+1} - \frac{v_i^l}{\rho_i^l} )$
independently in parallel. For indices $i \in \{1,2,\dots,s-1\}$, the
proximal maps are projections onto sets $\mathcal{D}$, $\mathcal{E}$ or
$\mathcal{F}$. These projections do not include linear operators, and we know the solutions in
closed form (e.g., $\ell_1$-norm, $\ell_2$-norm, rank, cardinality, bounds).

\textbf{$\rho_i^{l+1},\gamma_i^{l+1}$ updates}. We use the updating
scheme for $\rho$ and $\gamma$ from adaptive-relaxed ADMM, introduced by
\citet{Xu_2017_CVPR}. Numerical results show that this updating scheme
accelerates the convergence of ADMM
\citep{pmlr-v54-xu17a, Xu_2017_CVPR, pmlr-v70-xu17c}, also is also
robust when solving some non-convex problems \citep{empiricalncvxadmm}.
We use a different relaxation and penalty parameter for each function
$f_i(y_i)$, as do \citet{Song:2016:FAA:3015812.3015924};
\citet{pmlr-v70-xu17c}, which allows
$\rho_i$ and $\gamma_i$ to adapt to the various linear operators of
different dimensions that correspond to each constraint set.

\textbf{Parallelism and communication}. The only serial part of the
algorithm defined in~\eqref{PMinkowski_iters} is the $x^{l+1}$
computation. We use multi-threaded MVPs in the compressed diagonal
format if $Q$ has a banded structure. The other parts of the iterations, $y_i^{l+1}$,
$v_i^{l+1}, `\rho_i^{l+1}, \gamma_i^{l+1}$, are all independent so we
can compute them in parallel for each index $i$. There are two
operations in~\eqref{PMinkowski_iters} that require communication between
workers that carry out computations in parallel. We need to send
$x^{l+1}$ to every worker that computes a $y_i^{l+1}$, $v_i^{l+1}$,
$\rho_i^{l+1}$, and $\gamma_i^{l+1}$. The second and last piece of
communication is the map-reduce parallel sum to form the right-hand side
for the next iteration when we compute $x^{l+1} = \sum_{i=1}^s \Big( A_i^T(\rho_i^l y_i^{l} + v_i^l) \Big)$.

In practice, we will use the proposed algorithm to solve problems that often involve non-convex sets. Therefore, we do not provide guarantees that
algorithms like ADMM behave as expected, because their convergence
proofs typically require closed, convex and proper functions, see, e.g.,
\cite{Boyd:2011:DOS:2185815.2185816,eckstein2015understanding}. This is not a point of great concern
to us, because the main motivation to base our algorithm on ADMM is
rapid empirical convergence, ability to deal with many
constraint sets efficiently, and strong empirical performance in case of non-convex
sets.

\section{Formulation of inverse problems with generalized Minkowski
constraints}\label{formulation-of-inverse-problems-with-generalized-minkowski-constraints}

So far, we proposed a generalization of the Minkowski set ($\mathcal{M}$,
equation~\eqref{Minkowski_general}), and developed an algorithm to compute
projections onto this set. The next step is to solve inverse problems where
the generalized Minkowski set describes the prior knowledge. Therefore, we need to combine the set $\mathcal{M}$ with a data-fitting procedure, for which we discuss two formulations. One is primarily suitable
when the data-misfit function is computationally expensive to evaluate. The second formulation is for inverse
problems where the forward operator is both linear and computationally
inexpensive to evaluate. We discuss the two approaches in more detail
below.

\subsection{Inverse problems with computationally expensive data-misfit
evaluations}\label{inverse-problems-with-computationally-expensive-data-misfit-evaluations}

We consider a non-linear and possibly non-convex data-misfit function
$f(m) : \mathbb{R}^N \rightarrow \mathbb{R}$ that depends on model
parameters $m \in \mathbb{R}^N$. Our assumptions for this inverse
problem formulation is that the computational budget allows for much
fewer data-misfit evaluations than the required number of iterations to project
onto the generalized Minkowski set, as defined in \eqref{PMinkowski_iters}.
We can deal with this imbalance by attempting to make as much progress
towards minimizing $f(m)$, while always satisfying the constraints. The
minimization of the data-misfit, subject to satisfying the generalized
Minkowski constraint is then formulated as
\begin{equation}
\min_m f(m) \:\: \text{s.t.} \:\: m \in \mathcal{M}.
\label{nonlin_inv_prob}
\end{equation}
 If we solve this problem with algorithms that use a projection onto
$\mathcal{M}$ at every iteration, the model parameters $m$ satisfy the constraints at every iteration; a property desired by
several works in non-convex geophysical parameter estimation \citep{smithyman2015constrained, doi:10.1190/tle36010094.1, Esser2016arch, TVWRI2,ournewpreprint}. These works obtain better model reconstructions from non-convex problems
by carefully changing the constraints during the data-fitting procedure.
The first two numerical experiments in this work use the spectral
projected gradient algorithm \citep[SPG,][]{Birgin:1999:NSP:588891.589081,doi:10.1093/imanum/23.4.539}). At iteration $l$, SPG updates the model as
\begin{equation}
m^{l+1} = (1-\gamma) m^l - \gamma \mathcal{P}_{\mathcal{M}} (m^l - \alpha \nabla_{m}f(m^l)),
\label{SPG_iter2}
\end{equation}
 where $\mathcal{P}_{\mathcal{M}}$ is the Euclidean projection onto
$\mathcal{M}$. The Barzilai-Borwein \citep{BARZILAI01011988} step-length
$\alpha>0$ is a scalar approximation of the Hessian that is informed by
previous model estimates and gradients of $f(m)$. A non-monotone
line-search estimates the scalar $\gamma \in (0,1]$ and prevents $f(m)$
from increasing too many iterations in sequence. The line-search
back-tracks between two points in a convex set if $\mathcal{M}$ is
convex and the initial $m_0$ is feasible, so every line-search iterate
is feasible by construction. SPG thus requires a single projection onto
$\mathcal{M}$ if all constraint sets that construct the generalized Minkowski set are convex.

\subsection{Linear inverse problems with computationally cheap forward
operators}\label{linear-inverse-problems-with-computationally-cheap-forward-operators}

Contrary to the previous section, we now assume a linear relationship
between the model parameters $m \in \mathbb{R}^N$ and the observed data,
$d_\text{obs} \in \mathbb{R}^M$. The second assumption, for the problem formulation in this section, that the evaluation of the linear forward operator is not much more time
consuming than other computational components in the iterations from
\eqref{PMinkowski_iters}. Examples of such operators
$G \in \mathbb{R}^{M \times N}$ are masks and
blurring kernels. We may then put data-fitting and regularization on the
same footing and formulate an inverse problem with constraints as a
feasibility or projection problem. Both these formulations add a
data-fit constraint to the constraints that describe model properties, e.g., \citep{STE_2, STE_1, STEstimation, COMBETTES1996155}.
There are many data-fit constraint sets, including the point-wise set
$\mathcal{G}^{\text{data}} \equiv \{ m \: | \: l[i] \leq (G m - d_\text{obs})[i] \leq u[i]\}$
with lower and upper bounds on the misfit. We use the notation $l[i]$
for entry $i$ of the lower-bound vector $l$. The data-fit constraint can
be any set onto which we know how to project. An example of a global
data-misfit constraint is the norm-based set
$\mathcal{G}^{\text{data}} \equiv \{ m \: | \: \sigma_l \leq \| G m - d_\text{obs} \| \leq \sigma_u \}$
with scalar bounds $\sigma_l < \sigma_u$. This set is non-convex if
$\sigma_l > 0$, i.e., the annulus constraint in case of the $\ell_2$
norm. This set has a `hole' in the interior of the set that explicitly
avoids fitting the data noise in the $\ell_2$ sense.

We denote our formulation of a linear inverse problem with a data-fit constraint, and a generalized Minkowski set constraint \eqref{Minkowski_general} on the model parameters as
\begin{equation}
\min_{x} \frac{1}{2} \| x - m \|_2^2 \quad \text{s.t.} \quad \begin{cases} x \in \mathcal{M} \\
x \in \mathcal{G}^{\text{data}}
\end{cases}.
\label{LIP_Minkowski}
\end{equation}
 The solution is the projection of an initial guess, $m$, onto the
intersection of a data-fit constraint and a generalized Minkowski
constraint on the model parameters. As before, there are constraints on
the model $x$, as well as the components $u$ and $v$.
Problem~\eqref{LIP_Minkowski} is the same as before in
Equation~\eqref{Minkowski_general} and we can solve it with the
algorithm from the previous section. In the current case, we have
one additional constraint on the sum of the components.

\section{Numerical examples}\label{numerical-examples}
The algorithm and examples are available open-source at \url{https://petersbas.github.io/GeneralizedMinkowskiSetDocs/}. The Julia software is build on-top of the \texttt{SetIntersectionProjection}\footnote{\url{https://github.com/slimgroup/SetIntersectionProjection.jl}} package \citep{peters2019algorithms}. 

\subsection{Seismic full-waveform inversion
1}\label{seismic-full-waveform-inversion-1}

We start with a simple experiment and show how
a Minkowski set describes the provided prior knowledge naturally and results in a better model estimate compared to a single constraint set
or intersection of multiple sets. The problem is to estimate the
acoustic velocity $m \in \mathbb{R}^N$ of the model in
Figure~\ref{Fig:FWI-results} from observed seismic data modeled by the
Helmholtz equation. This problem, known as full-waveform inversion \citep[FWI,][]{TarantolaA,Pratt98,virieux09}, is often
formulated as the minimization of a differentiable, but non-convex
data-fit
\begin{equation}
f(m) = \frac{1}{2} \| d_\text{predicted}(m) - d_\text{observed} \|_2^2,
\label{FWI_obj}
\end{equation}
 where the partial-differential-equation constraints are already
eliminated and are part of $d_\text{predicted}(m)$, see, e.g.,
\citet{Haber2000}. The observed data, $d_\text{observed}$ are discrete
frequencies of $\{3.0,6.0,9.0\}$ Hertz.

Figure~\ref{Fig:FWI-results} shows the true model, initial guess for
$m$, and the source and receiver geometry. We assume prior information
about the bounds on the parameter values, and that the anomaly has a
rectangular shape with a lower velocity than the background.

The results in Figure~\ref{Fig:FWI-results} using bounds or bounds and
the true anisotropic total-variation (TV) as a constraint, do not lead
to a satisfying model estimate. The diagonally shaped model
estimates are mostly due to the source and receiver positioning, known
as vertical seismic profiling (VSP) in geophysics. Here we will show that the generalized Minkowski set $\mathcal{M}$ \eqref{Minkowski_general} leads to an improved
model estimate using convex sets only. If we have the prior
knowledge that the anomaly we need to find has a lower velocity than the
background medium, we can easily include this information as a
Minkowski set. The following four sets summarize our prior knowledge:

\begin{enumerate}
\def\labelenumi{\arabic{enumi}.}
\itemsep1pt\parskip0pt\parsep0pt
\item
  $\mathcal{F}_1 = \{ x \: | \: 2350 \leq x[i] \leq 2550 \}$ : bounds on sum
\item
  $\mathcal{F}_2 = \{ x \: | \: \| ( (D_z \otimes I_x)^T \:\: (I_z \otimes D_x)^T )^T x \|_1 \leq \sigma \}$ : anisotropic total-variation on sum
\item
  $\mathcal{D}_1 = \{ u \: | \: -150 \leq u[i] \leq 0 \}$ : bounds on anomaly
\item
  $\mathcal{E}_1 = \{ v \: | \: v[i] = 2500 \}$ : bounds on background
\end{enumerate}

The generalized Minkowski set combines the four above sets as $( \mathcal{F}_1 \bigcap \mathcal{F}_2 ) \bigcap (\mathcal{D}_1 + \mathcal{E}_1)$.
In words, we fix the background velocity, require any anomaly to be
negative, and the total model estimate has to satisfy bound constraints
and have a low anisotropic total-variation. To minimize the data-misfit
subject to the generalized Minkowski constraint,
\begin{equation}
\min_m \frac{1}{2} \| d_\text{predicted}(m) - d_\text{observed} \|_2^2 \:\: \text{s.t.} \:\: m \in ( \mathcal{F}_1 \bigcap \mathcal{F}_2 ) \bigcap (\mathcal{D}_1 + \mathcal{E}_1),
\label{FWI_obj_mink}
\end{equation}
 we use the SPG algorithm as described in the previous section, with
$15$ iterations and a non-monotone line search with a memory of five
function values. The result that uses the generalized Minkowski
constraint (Figure~\ref{Fig:FWI-results}e) is much better compared to
bounds and the correct total-variation because the constraints on the
sign of the anomaly prevent incorrect high-velocity artifacts.

\begin{figure}
\centering
\captionsetup[subfigure]{labelformat=empty}
\subfloat[]{\includegraphics[width=0.330\hsize]{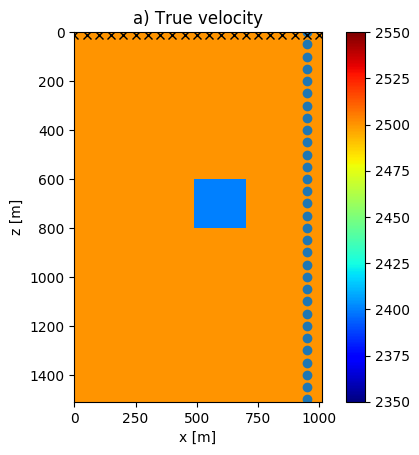}}
\subfloat[]{\includegraphics[width=0.330\hsize]{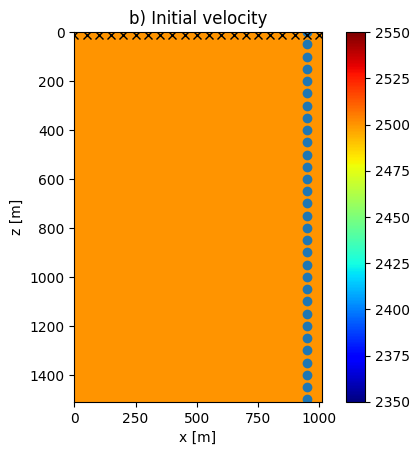}}
\\
\subfloat[]{\includegraphics[width=0.330\hsize]{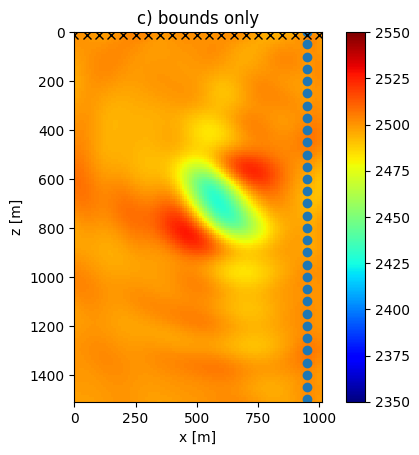}}
\subfloat[]{\includegraphics[width=0.330\hsize]{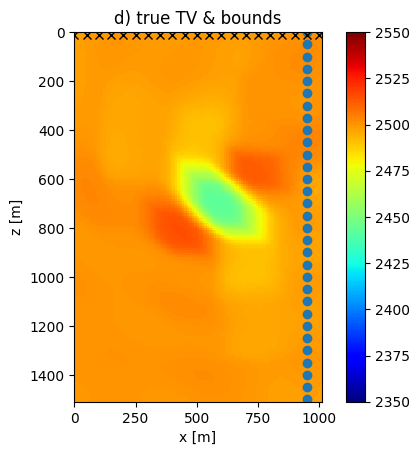}}
\subfloat[]{\includegraphics[width=0.330\hsize]{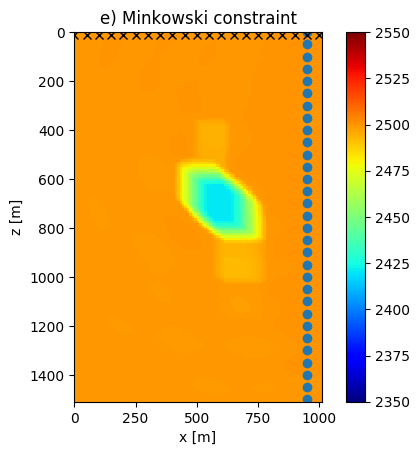}}
\caption{The true model for the data generation for the full-waveform
inversion $1$ example, the initial guess for parameter estimation, and
the model estimates with various constraints. Crosses and circles
indicate receivers and sources, respectively.}\label{Fig:FWI-results}
\end{figure}

While there are other ways to fix a background model and invert for an
anomaly, this example illustrates that our proposed regularization
approach incorporates information on the sign of an anomaly conveniently
and the constraints remain convex. It is straightforward to change and
add constraints on each component, also in the more realistic situation
that the background is not known and should not be fixed, as we show in
the following example.

\subsection{Seismic full-waveform inversion
2}\label{seismic-full-waveform-inversion-2}

This time, the challenge is to estimate a model
(Figure~\ref{Fig:FWI-results2}a) that has both a background and an
anomaly component that are very different from the initial guess
(Figure~\ref{Fig:FWI-results2}b). This means we can no longer fix one
of the two components of the generalized Minkowski sum.

The experimental setting is a bit different from the previous example.
The sources are in one borehole, the receivers in another borehole at
the other side of the model (cross-well full-waveform inversion). Except for a single high-contrast anomaly, the velocity is increasing monotonically, both gradually and discontinuously. The prior
knowledge that we assume consists of \emph{i)} upper and lower bounds on the
velocity and also on the anomaly \emph{ii)} the model is relatively simple in the sense that we
assume it has a rank of at most five \emph{iii)} the background
parameters are increasing monotonically with depth \emph{iv)} the
background is varying smoothly in the lateral direction \emph{v)} the
size of the anomaly is not larger than one fifth of the height of the
model and not larger than one third of the width of the model. We do not
assume prior information on the total-variation of the model, but for
comparison, we show the result when we use the true total-variation as a constraint. The following sets formalize the aforementioned prior knowledge:

\begin{enumerate}
\def\labelenumi{\arabic{enumi}.}
\itemsep1pt\parskip0pt\parsep0pt
\item
  $\mathcal{F}_1 = \{ x \: | \: 2350 \leq x[i] \leq 2850 \}$
\item
  $\mathcal{F}_2 = \{ x \: | \: \| ( (D_z \otimes I_x)^T \: (I_z \otimes D_x)^T )^T x \|_1 \leq \sigma \}$
\item
  $\mathcal{F}_3 = \{ x \: | \: \operatorname{rank}(x) \leq 5 \}$
\item
  $\mathcal{D}_1 = \{ x \: | \: 2350 \leq x[i] \leq 2850 \}$
\item
  $\mathcal{D}_2 = \{ u \: | \: 0 \leq (D_z \otimes I_x) u \leq \infty \}$
\item
  $\mathcal{D}_3 = \{ u \: | \: -0.1 \leq (I_z \otimes D_x) u \leq 0.1 \}$
\item
  $\mathcal{E}_1 =  \{ v \: | \: 300 \leq v[i] \leq 350 \}$
\item
  $\mathcal{E}_2 = \{ v \: | \: \operatorname{cardinality}(v) \leq (n_z / 5 \times n_x/3) \}$
\end{enumerate}

As before, the sets $\mathcal{F}_k$ act on the sum of components,
$\mathcal{D}_i$ describe component one (background), and $\mathcal{E}_j$
constrain the other component (anomaly). This collection of prior knowledge is quite specific, but we use it to illustrate how a generalized Minkowski set conveniently includes all this information. Figure~\ref{Fig:FWI-results2}c
shows the model $m$ found by SPG applied to the problem
$\min_m f(m) \:\: \text{s.t.} \:\: m \in \mathcal{F}_1$. We see
oscillatory features in the result with bound constraints only, but the
main issue is the appearance of a low-velocity artifact, located just below the true anomaly. Figure~\ref{Fig:FWI-results2}d shows that even if we know the correct total-variation, the result is
less oscillatory than using just bound constraints, but still shows an
erroneous low-velocity anomaly. When we also include the rank
constraint, i.e., we use the set
$\mathcal{F}_1 \bigcap \mathcal{F}_2 \bigcap \mathcal{F}_3$, the result
does not improve (Figure~\ref{Fig:FWI-results2}e). The generalized
Minkowski set
$\big( \bigcap_{k=1}^3 \mathcal{F}_k \big) \bigcap \big(\bigcap_{i=1}^3 \mathcal{D}_i + \bigcap_{j=1}^2 \mathcal{E}_j \big)$
does not yield a result with the large incorrect low-velocity artifact
just below the correct high-velocity anomaly (Figure \ref{Fig:FWI-results2}g),
even though we did not include information on the sign of the anomaly as
we did in the previous example. There are still two smaller horizontal
and vertical artifacts. Overall, the Minkowski set based constraint
results in the best background and anomaly estimation.

\begin{figure}
\centering
\captionsetup[subfigure]{labelformat=empty}
\subfloat[]{\includegraphics[width=0.330\hsize]{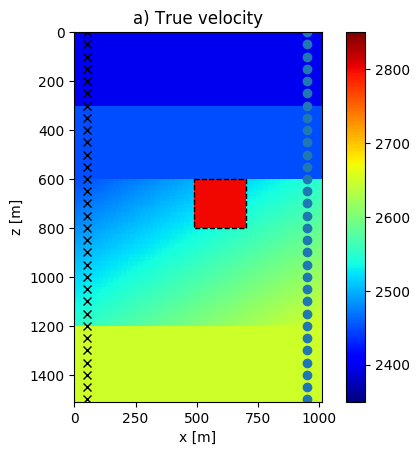}}
\subfloat[]{\includegraphics[width=0.330\hsize]{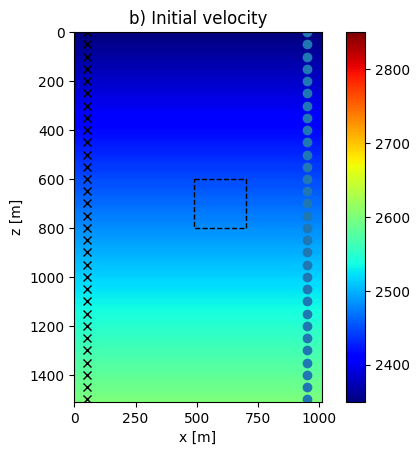}}
\subfloat[]{\includegraphics[width=0.330\hsize]{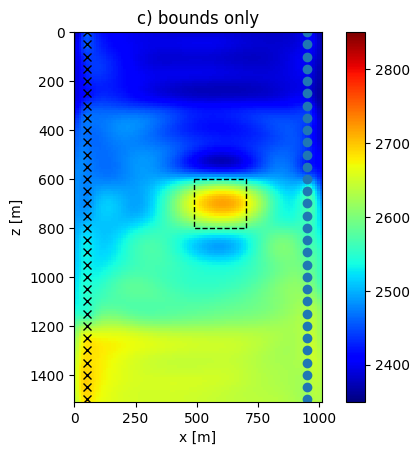}}
\\
\subfloat[]{\includegraphics[width=0.330\hsize]{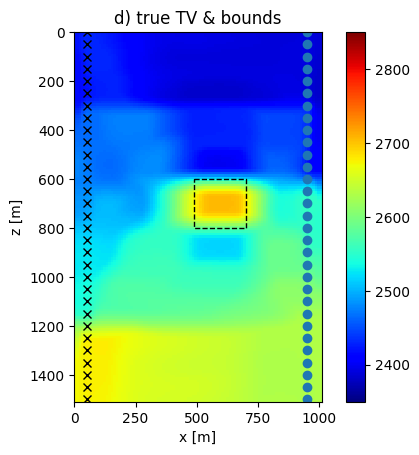}}
\subfloat[]{\includegraphics[width=0.330\hsize]{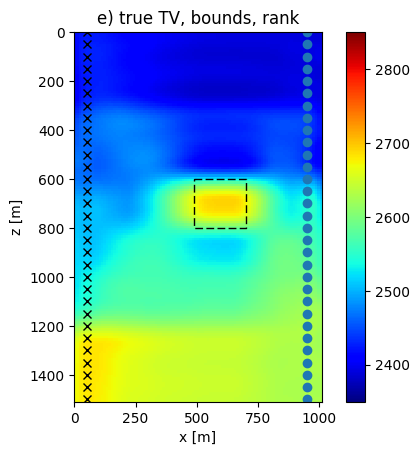}}
\subfloat[]{\includegraphics[width=0.330\hsize]{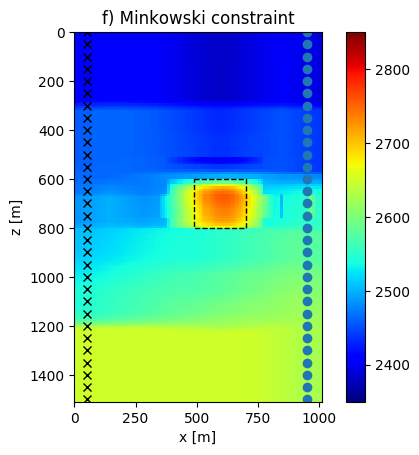}}
\caption{The true and initial models corresponding to the full-waveform
inversion $2$ example. Figure shows parameter estimation results with
various intersections of sets, as well as the result using a generalized
Minkowski constraint set. Only the result obtained with the generalized
Minkowski set does not show an incorrect low-velocity
anomaly.}\label{Fig:FWI-results2}
\end{figure}

This example shows that the generalized Minkowski set allows for the
inclusion of prior knowledge on the two (or more) different components,
as well as their sum. The results show that this leads to improved model
estimates. Information that we may have about background or anomaly
is often difficult or impossible to include in an inverse problem as a
single constraint set or intersection of multiple sets, but it fits the summation structure of the generalized Minkowski set. In
many practical problems, we do have some information about an anomaly.
When looking for air or water filled voids and tunnels in engineering or
archeological geophysics, we know that the acoustic wave propagation
velocity is usually lower than the background and we also have at least
a rough idea about the size of the anomaly. In seismic hydrocarbon
exploration, there are high-contrast salt structures in the subsurface,
almost always with higher acoustic velocity than the surrounding
geology.

\subsection{Video processing}\label{video-processing}

Background-anomaly separation of a tensor $T \in \mathbb{R}^{n_x \times n_y \times n_t}$, where $x$ and $y$ are
the two spatial coordinates and $t$ is the time, is a common problem in video processing. The separation problem
is often used to illustrate robust principal component analysis (RPCA),
and related convex and non-convex formulations of sparse + low-rank
decomposition algorithms
\citep[e.g.,][]{Candes:2011:RPC:1970392.1970395, NIPS2014_5430, rpca_ncvx_approx, arXiv170202241D_2017}.

Here we show that the generalized Minkowski set for an
inverse problem, proposed in Equation~\eqref{Minkowski_general}, describes prior knowledge that is obtained from an example, as well as simple human intuition about a background-anomaly separation problem in video
processing. To include multiple pieces of prior knowledge,
we choose to work with the video in tensor format and use the
flexibility of our regularization framework to impose constraints on the
tensor, as well as on individual slices and fibers. This is different
from RPCA approaches that matricize or flatten the video tensor to a
matrix of size $n_x n_y \times n_t$, such that each column of the matrix
is a vectorized time-slice
\citep{Candes:2011:RPC:1970392.1970395, NIPS2014_5430, rpca_ncvx_approx},
and also differs from tensor-based RPCA methods that work with a tensor
only, e.g., \citep{Zhang_2014_CVPR, Wang_2015_ICCV_Workshops}. 

Beyond the basic decomposition problem, the escalator video comes with
some additional challenges. There is a dynamic background component (the steps of the escalator), and there are reflections of people in the glass that
are weak anomalies and duplicates of persons. The video contains noise
and part of the background pixel intensity changes significantly ($55$
on a $0-255$ grayscale) over time. We subtract the mean of each time frame as
a pre-processing step to mitigate the change in intensity. Below, we
describe simple methods to derive prior knowledge for the video background component and how to translate observations and human intuition to a generalized Minkowski set.

\textbf{Constraint sets for background.} We assume that the last $20$ time frames, that do not contain people, are available to derive constraints for the background. From these frames, we use the minimum and maximum value
for each pixel over time as the bounds for the background component,
denoted as set $\mathcal{D}_1$. This is one constraint defined per time fiber. The second constraint is the subspace
spanned by the last $20$ frames. We require that each time frame of the
background be a linear combination of the training frames organized as a
matrix $S \in \mathbb{R}^{n_x n_y \times 20}$, where each column is a
vectorized video frame of $T$. We denote this constraint as
$\mathcal{D}_2 = \{ u \: | \: u = S c, \: c \in \mathbb{R}^{20} \}$, with
coefficient vector $c$, which we obtain during the projection operation:
$\mathcal{P}_{\mathcal{D}_2}(u) = S(S^T S)^{-1}S^T u$. After computing
the singular value decomposition $S=U \Sigma V^T$, the projection
simplifies to $\mathcal{P}_{\mathcal{D}_2}(u) = U U^T u$

\textbf{Constraint sets for sum of components.} We constrain the sum of
the background and anomaly components to the interval of grayscale
values $[0 - 255 ]$ minus the mean of each time-frame, denoted as set
$\mathcal{F}_1$.

\textbf{Constraint sets for anomaly.} There are bound constraints on the anomaly component that we define as the bounds
on the sum minus the bounds on the background. We denote this set as $\mathcal{E}_1$. To enhance the quality of
the anomaly component, we add various types of sparsity constraints. Non-convex sets and set definitions per tensor slice translate simple intuition into constraints. The first type of sparsity constraint is the set
$\mathcal{E}_2 = \{ T \: | \: \operatorname{card}(T_{\Omega_i}) \leq (n_x/4 \times n_y/4) \: \forall i \in \{1,2,\dots,n_t\} \}$
where $T_{\Omega_i}$ is a time slice of the video tensor. This
constraint limits the number of anomaly pixels in each frame to $1/16$
of the total number of pixels in each time slice. The second and third
constraint sets are limits on the vertical and horizontal derivative of
each time-frame image separately. If we assume the prior knowledge that
there are no more than ten persons in the video at each time, we can use
$\mathcal{E}_3 = \{ T \: | \: \operatorname{card}((I_x \otimes D_y) \operatorname{vec}(T_{\Omega_i})) \leq 480, \: i \in \{1,2,\dots,n_t\}\}$,
based on the rough estimate of
$10 \: \text{persons} \times 12 \: \text{pixels wide} \times 4 \: \text{boundaries}$
(the four vertical boundaries are background - head - upper body - legs
- background). Similarly for the horizontal direction, we define
$\mathcal{E}_4 = \{ T \: | \: \operatorname{card}((D_x \otimes I_y) \operatorname{vec}(T_{\Omega_i})) \leq 440, \: i \in \{1,2,\dots,n_t\}\}$.

Putting it all together, we project the video onto the generalized
Minkowski set, i.e., we solve
\begin{equation}
\min_x \frac{1}{2} \| x - \operatorname{vec}(T) \|_2^2 \:\: \text{s.t.} \:\: x \in \mathcal{F}_1 \bigcap \big(\bigcap_{i=1}^2 \mathcal{D}_i + \bigcap_{j=1}^4 \mathcal{E}_j \big)
\label{video_prob}
\end{equation}
 using the iterations derived in \eqref{PMinkowski_iters}. Our
formulation implies that the projection of a vector is always the sum of
the two components, but this does not mean that $x$ is equal to $T$ at
the solution, because we did not include a constraint on $x$ that says
we need to fit the data accurately. We did not use a data-fit constraint
because it is not evident how tight we want to fit the data
or how much noise there is. By computing the projection of the original
video, we still include a sense of proximity to the observed data.

\begin{figure}
\centering
\includegraphics[width=1.000\hsize]{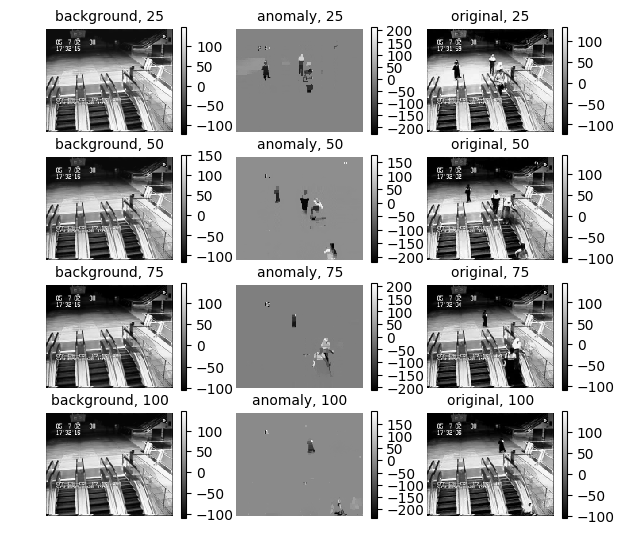}
\caption{Results of the generalized Minkowski decomposition applied to
the escalator video. The figure shows four frames. The most noticeable
artifacts are in the time stamp.}\label{Fig:escalator}
\end{figure}

The result of the generalized Minkowski decomposition of the video shown
in Figure~\eqref{Fig:escalator}, is visually better than the six methods compared by \citet{arXiv170202241D_2017}. The compared results often
show blurring of the escalator steps in the estimated background,
especially when a person is on the escalator. Several results also show
visible escalator structure in the anomaly estimation. Our simple
approach does not suffer from these two problems. We do not need to
estimate any penalty or trade-off parameters but rely on constraint sets based on intuition or whose parameters we can observe directly from a few training frames.  We were able to conveniently mix constraints on
slices and fibers of the tensor by working with the constrained
formulation.

\section{Discussion}\label{discussion}

So far, we described the concepts and algorithms for the case of a Minkowski sum with only two components. Our approach can handle more than two components, but the linear systems in Equation~\eqref{system-mat} will become larger. A
better solver than plain conjugate-gradients can mitigate increased
solutions times due to larger linear systems, possibly by taking the block structure into account. 

Another potential issue related to Minkowski sums of more than two components is that it will be less intuitive what type of solutions are in the Minkowski sum of sets. We can regain some intuition about the generalized Minkowski set by looking at sampled elements from the set. Samples are simply obtained by projecting vectors (possible solutions, reference models, random vectors, \ldots{}) onto the target set.

\section{Conclusions}\label{conclusions}

Inverse problems for physical parameter estimation and image and video processing often encounter model parameters with complex structure, so that it is difficult to describe the expected model parameters with a single set or intersection of multiple sets. In these situations, it may be easier to work with an additive model description where the model is the sum of morphologically distinct components.

We presented a regularization framework for inverse problems with the Minkowski set at its core. The additive structure of the Minkowski set allows us to enforce prior knowledge in the form of separate constraints on each of the model components. In that sense, our work differs from current approaches that rely on additive penalties for each component. As a result, we no longer need to introduce problematic trade-off parameters.

Unfortunately, the Minkowski set by itself is not versatile enough for physical parameter estimation because we also need to enforce bound constraints and other prior knowledge on the sum of the two components to ensure physical feasibility. Moreover, we would like to use more than one constraint per component to incorporate all prior knowledge that we may have available.

To deal with this situation, we proposed a generalization of the Minkowski set by defining it as the intersection of a Minkowski set with another intersection of constraint sets on the sum of the components. Each of the two or more components of the Minkowski set is also an intersection of sets. With this construction, we can enforce
multiple constraints on the model parameters, as well as multiple
constraints on each component. The generalized Minkowski set is convex if all contributing sets are convex. 

To solve inverse problems with these constraints, we discuss how to project onto generalized Minkowski sets based on the alternating direction method of multipliers. The projection enables projected-gradient based method to minimize nonlinear functions subject to constraints. We also showed that for linear inverse problems, the linear forward operator fits in the projection computation directly as a data-fitting constraint. This makes the inversion faster if the application of the forward operator does not take much time.

Numerical examples show how the generalized Minkowski set helps to solve non-convex seismic parameter estimation problems and a background-anomaly separation task in video processing, given prior knowledge on the model parameters, as well as on the components. The proposed regularization framework thus combines the benefits of constrained problem formulations and additive model descriptions. The algorithms and accompanying software are suitable for toy problems, as well as for larger problems such as videos on 3D grids. 

\bibliographystyle{abbrvnat}
\bibliography{bib_bas,bib_bas_mink}

\end{document}